\documentclass[a4paper,UKenglish]{article}
\usepackage{a4wide}
\usepackage{hyperref}
\hypersetup{
    colorlinks   = true,
    urlcolor     = black,
    linkcolor    = black,
    citecolor    = black
}

\title{Free constructions and coproducts of d-frames}
\author{Tom\'a\v s Jakl \and Achim Jung}
\date{April 7, 2017}

\usepackage{xparse}   % for NewDocumentCommand etc.
\usepackage{enumitem} % for axioms environment
\usepackage{stmaryrd} % for \llbracket and \rrbracket

\usepackage{scalerel}
\usepackage{stackengine}

\usepackage{amsfonts,amssymb, amsthm, amsmath}
\usepackage{tikz}
\usepackage{tikz-cd}
\usepackage{bussproofs}
\usepackage{dashbox}  % for higlightChange

\EnableBpAbbreviations

\usetikzlibrary{
    arrows,
    fit,
    calc,
    matrix,
    calc,
    decorations.pathreplacing,
    decorations.markings,
    intersections,
    shapes,
    snakes
}
\tikzset{commutative diagrams/.cd}
\tikzset{commutative diagrams/row sep/normal=1.3cm}
\tikzset{commutative diagrams/column sep/normal=1.3cm}
\tikzstyle{nattransf} = [-implies,double, double equal sign distance, shorten >=10pt, shorten <=10pt]
\tikzstyle{nattdiag} = [-implies,double, double equal sign distance]
\tikzstyle{rel} = [-,dashed,swap]

\usetikzlibrary{shapes}
\usetikzlibrary{shapes.geometric}
\usetikzlibrary{backgrounds}
\usetikzlibrary{patterns}

\usetikzlibrary{fadings}
\usetikzlibrary{shadows.blur}

\theoremstyle{plain}
\newtheorem{theorem}{Theorem}
\newtheorem{proposition}[theorem]{Proposition}
\newtheorem{lemma}[theorem]{Lemma}
\newtheorem{observation}[theorem]{Observation}
\newtheorem{corollary}[theorem]{Corollary}

\theoremstyle{definition}
\newtheorem*{remark}{Remark}
\newtheorem*{notation}{Notation}
\newtheorem*{definition}{Definition}

\newcommand\p[1]{\ensuremath{\mathcal{ #1 }}}
\newcommand\qtq[1]{\quad\text{#1}\quad}
\newcommand{\ttimes}{\mathord{\times}}
\newcommand\upset  {\ensuremath{\mathord{\uparrow}\mkern1mu}}
\newcommand\downset{\ensuremath{\mathord{\downarrow}\mkern1mu}}
\newcommand\operatorupX[1]{\,\ThisStyle{\ensurestackMath{%
  #1\stackengine{-0pt}{\,}{\SavedStyle\!^{\mathord{\uparrow}}}{O}{l}{F}{T}{S}}}\!}
\makeatletter
\newcommand\operatorup[1]{\!\mathop{\operatorupX{#1}}\@ifnextchar\{{\,}{\@ifnextchar[{\,}{\@ifnextchar({\,}{}}}}
\makeatother
\newcommand{\dirsqcup}{\operatorup{\bigsqcup}}
\newcommand{\dircup}  {\operatorup{\bigcup}}
\newcommand{\dirvee}  {\operatorup{\bigvee}}
\newcommand\dirsubseteq{\mathbin{\subseteq\!\dirup}}
\newcommand\dirup{\!{}^{\upset}}                    % OBSOLETE
\newcommand\finsubseteq{\mathbin{\subseteq\!^\text{fin}}}

\newcommand{\tot}{\ensuremath{\mathsf{tot}}}
\newcommand{\con}{\ensuremath{\mathsf{con}}}
\newcommand{\dtt}{\ensuremath{\text{\emph{t}}\mkern-3mu\text{\emph{t}}}}
\newcommand{\dff}{\ensuremath{\text{\emph{ff}}}}
\newcommand{\DEFEQ}                  {\overset{\text{\tiny def}}{\equiv}}

\newcommand{\sem}[1][-]{\ensuremath{\llbracket #1 \rrbracket}}

\newcommand\Di{\ensuremath{\p D^{\text{ind}}}}

\newcommand{\conO}{\ensuremath{\con_{1}}}
\newcommand{\totO}{\ensuremath{\tot_{1}}}

\newcommand{\conW}{\ensuremath{\con_{\wedge}}}
\newcommand{\conV}{\ensuremath{\con_{\vee}}}
\newcommand{\conWV}{\ensuremath{\con_{\wedge, \vee}}}
\newcommand{\conVW}{\conWV}
\newcommand{\totW}{\ensuremath{\tot_{\wedge}}}
\newcommand{\totV}{\ensuremath{\tot_{\vee}}}
\newcommand{\totWV}{\ensuremath{\tot_{\wedge, \vee}}}

\newcommand{\conWbV}{\ensuremath{\con_{\wedge, \bigvee}}}
\newcommand{\conVbW}{\ensuremath{\con_{\vee,   \bigwedge}}}
\newcommand{\conWVbV}{\ensuremath{\con_{\wedge, \vee, \bigvee}}}
\newcommand{\conWVbW}{\ensuremath{\con_{\wedge, \vee, \bigwedge}}}

\newcommand{\higlightChange}[1]{%
    {\color{gray}\dbox{\normalcolor$#1$}}%
}

\newcommand\drawline[4]{
  \pgfline{\pgfpoint{#1}{#2}}{\pgfpoint{#3}{#4}}
}
\pgfdeclarepatternformonly{pluses}% name
  {\pgfpointorigin}       % lower left of bounding box
  {\pgfqpoint{7pt}{7pt}}% upper right of bounding box
  {\pgfqpoint{6pt}{6pt}}% step vector
  {% code
      \pgfsetlinewidth{0.3pt}
      \drawline{1.5pt}{0.0pt}{1.5pt}{3.0pt}
      \drawline{0.0pt}{1.5pt}{3.0pt}{1.5pt}
      \drawline{4.5pt}{3.0pt}{4.5pt}{6.0pt}
      \drawline{3.0pt}{4.5pt}{6.0pt}{4.5pt}
      \pgfusepath{stroke}
  }

\pgfdeclarepatternformonly{minuses}
  {\pgfpointorigin}
  {\pgfqpoint{7pt}{7pt}}
  {\pgfqpoint{6pt}{6pt}}%
  {
      \pgfsetlinewidth{0.4pt}
      \drawline{0.0pt}{5.5pt}{3.0pt}{5.5pt}
      \drawline{3.0pt}{2.5pt}{6.0pt}{2.5pt}
      \pgfusepath{stroke}
  }

\newcommand\Fr[1]{\ensuremath{\mathbb Fr\!\left< #1 \right>}}

\newcommand\CIdl[1][\pm]{\ensuremath{\p C_{#1}\text{-}\mathtt{Idl}\!\left(B_{#1}\right)}}
\newcommand\CIdlG[2][\pm]{\ensuremath{\p C_{#1}\text{-}\mathtt{Idl}\!\left<#2\right>}}

\newcommand\CON[1][\conO]{\ensuremath{\mathtt{CON}\langle #1 \rangle}}
\newcommand\TOT[1][\totO]{\ensuremath{\mathtt{TOT}\langle #1 \rangle}}
\newcommand\sCON{\ensuremath{\p D(\downset \conWV)}}

\newcommand\dL{\ensuremath{(L_+, L_-;\, \con, \tot)}}
\newcommand\dPres{\ensuremath{(B_+, B_-;\, \p C_+, \p C_-;\, \conO, \totO)}}
\newcommand\dB{\ensuremath{(B_+, B_-;\, \conO, \totO)}}
\newcommand\dBC[1][\pm]{\ensuremath{(B_{#1}, \p C_{#1})}}
\newcommand\dLGen{\ensuremath{(L_+, L_-;\, \CON, \TOT)}}

\newcommand\ones{\ensuremath{\overline 1}}
\newcommand\npm{\ensuremath{\mathbf{n_\pm}}}
\newcommand\np{\ensuremath{\mathbf{n_+}}}
\newcommand\nm{\ensuremath{\mathbf{n_-}}}

\newcommand\OLxOL{\ensuremath{\bigoplus_i L_+^i\ttimes \bigoplus_i L_-^i}}

\NewDocumentEnvironment{axioms}{O{Ax}}{\enumerate[label=(#1-\expandafter\arabic*), labelindent=5em, leftmargin=*, labelsep=1.5em]}{\endenumerate}
\newenvironment{diagram}{\begin{center}\begin{tikzcd}}{\end{tikzcd}\end{center}}

\begin{document}
\maketitle

\begin{abstract}
    A general theory of presentations for d-frames does not yet exist. We review the difficulties and give sufficient conditions for when they can be overcome. As an application we prove that the category of d-frames is closed under coproducts.
\end{abstract}

\section{Introduction}

In his celebrated \emph{Domain Theory in Logical Form}~\cite{abramsky87a}, Abramsky describes a flexible framework for connecting the denotational semantics of a programming language with an algebraic presentation of a program logics. The denotational spaces are spectral spaces and the algebras are distributive lattices; they are connected via Stone duality~\cite{jung13a}.

The attempt to expand the scope of Abramsky's work to cover probabilistic and real-number computation led to the study of stably compact spaces and their Stone duality. Later, it was shown that stably compact spaces have a very natural bitopological description; they are exactly the compact regular bitopological spaces (or \emph{bispaces} for short)~\cite{jung06,lawson2011stably}. Moreover, the Stone-type duality between bispaces and \emph{d-frames} given in \cite{jung06} has a finitistic description in the compact regular case, and so we can try to extend Abramsky's work to this setting.

Free constructions of distributive lattices are an essential tool of Domain Theory in Logical Form and therefore a general theory of free constructions of d-frames is highly desirable. Unfortunately, no such theory exists as of yet. The difficulty lies in the mixed algebraic-relational nature of d-frames and particularly in the axiom (\con-\tot). In the absence of a general theory one can look at special instances of the problem where the difficulties with (\con-\tot) can be controlled. This is our approach in this paper.

The carrier of a d-frame is two-sorted, consisting of two standard frames $L_+$ and~$L_-$. We use the usual generator and relations machinery to present them separately. The remaining parts of the structure, the consistency and totality relations $\con, \tot \subseteq L_+\ttimes L_-$ can be specified by generating relations \conO{} and \totO{}, but it is not clear how to make sure that (\con-\tot), the only axiom that bonds both relations, will hold in the generated structure. Rather than solve this general problem, we provide sufficient conditions which can be checked in an early stage of the generating process (Section~\ref{s:d-frames}).

As an application, we prove that the category of d-frames is closed under coproducts (Section~\ref{s:coproducts}). In forthcoming work, \cite{jakljung2017}, we show that the same techniques allow us to define the d-frame that corresponds to the \emph{Vietoris power space} over a bispace. Together, this lays the foundation for the extension of Abramsky's program as explained above as well as four-valued coalgebraic logic (inspired by~\cite{klin2007coalgebraic,kupkekurzvenema2004stone,jacobs2001many}). We believe that our results are interesting from a model-theoretic perspective as well, as they provide an example of a free construction for a two-sorted algebraic-relational structure. Although not completely general, our techniques hold promise for extending many other frame-theoretic constructions to d-frames (for examples see \cite{johnstone82,picadopultr2011frames,vickers89}).

\section{Preliminaries}

Frames are algebraic structures which capture the order-theoretic properties of the lattice of open sets of a topological space. We say that a complete lattice $(L;\, \bigvee, \wedge, 0, 1)$ is a \emph{frame} if it satisfies the following infinitary distributivity law
\begin{axioms}
    \item[(Frm)] $b \wedge (\bigvee_i\, a_i) = \bigvee_i\, (b \wedge a_i)$.
\end{axioms}
The counterparts to continuous maps are the \emph{frame homomorphisms} which are maps distributing over \emph{all} joins and \emph{all finite} meets.

A topological space $(X;\, \tau)$ gives rise to a frame: the lattice of its open sets ordered by set inclusion $\Omega(X) = (\tau;\, \bigcup, \cap;\, \emptyset, X)$ is a frame. Also, any continuous map $f\colon X\to Y$ gives rise to a frame homomorphism $\Omega(f)\colon \Omega(Y) \to \Omega(X)$ as $U\in \tau^Y \mapsto f^{-1}[U] \in \tau^X$.

\medskip

Following the example of frames we have d-frames as the algebraic counterparts to bitopological spaces (or \emph{bispaces}, for short)\footnote{Bispaces are the structures $(X;\, \tau_+, \tau_-)$ where $(X;\, \tau_+)$ and $(X; \tau_-)$ are topological spaces. A map between two bispaces $f\colon X\to Y$ is \emph{bicontinuous} if both $f_+\colon (X;\, \tau_+^X) \to (Y;\, \tau_+^Y)$ and $f_-\colon (X;\, \tau_-^X) \to (Y;\, \tau_-^Y)$ (which are acting on the underlying set $X$ the same way as $f$ does) are continuous.}. Because bispaces have two topologies, we expect to have two frames, $L_+$ and $L_-$, as part of the structure of d-frames.

This alone has some consequences. We can recognise two orders in the product $L_+\ttimes L_-$; the first is the \emph{information order} $\sqsubseteq$ where, for $\alpha = (\alpha_+, \alpha_-)$ and $\beta = (\beta_+, \beta_-) \in L_+\ttimes L_-$, $\alpha \sqsubseteq \beta$ iff $\alpha_+ \leq \beta_+$ and $\alpha_- \leq \beta_-$. The second is the \emph{logical order} $\leq$ where $\alpha \leq \beta$ iff $\alpha_+ \leq \beta_+$ and $\alpha_- \geq \beta_-$. Both $(L_+\ttimes L_-; \sqsubseteq)$ and $(L_+\ttimes L_-; \leq)$ are bounded distributive lattices with meets and joins computed as follows
\begin{align*}
    \alpha\vee   \beta = (\alpha_+\vee   \beta_+, \alpha_-\wedge \beta_-), &\qquad
    \alpha\sqcup \beta = (\alpha_+\vee   \beta_+, \alpha_-\vee   \beta_-), \\
    \alpha\wedge \beta = (\alpha_+\wedge \beta_+, \alpha_-\vee   \beta_-), &\qquad
    \alpha\sqcap \beta = (\alpha_+\wedge \beta_+, \alpha_-\wedge \beta_-).
\end{align*}
The smallest and largest elements in the information order are $\bot = (0,0)$ and $\top = (1,1)$, and in the logical order $\dff = (0,1)$ and $\dtt = (1,0)$, respectively.

With just two frames we would not be able to express many bitopological properties. One can require $L_+$ and $L_-$ to be subframes of a bigger frame representing the join of the two topologies as proposed by Banashewski~\cite{bbh1983biframes}. Or, following the second author and Moshier~\cite{jung06}, we can require two binary relations \con{} and \tot{} between the two frame components where $(a,b)\in \con$ corresponds to $a$ being disjoint from $b$, and $(a,b)\in \tot$ if $a$ and $b$ cover the whole space. Our work takes the second approach.

Formally, then, a \emph{d-frame} is a structure $\p L = \dL$ such that $L_+$ and $L_-$ are frames and the binary \emph{consistency} $\con \subseteq L_+\ttimes L_-$ and \emph{totality} $\tot \subseteq L_+\ttimes L_-$ relations satisfy the following axioms, for all $\alpha, \beta\in L_+\ttimes L_-$:
\begin{axioms}
    \item[(\con--$\downset$)] $\alpha \in \con$ and  $\beta \sqsubseteq \alpha \implies \beta \in \con$,

    \item[(\tot--$\upset$)] $\alpha \in \tot$ and $\beta \sqsupseteq \alpha \implies \beta \in \tot$,

    \item[(\con,\tot--$\dtt,\dff$)] $\dtt \in \con$ and $\dtt \in \tot$, $\dff \in \con$ and $\dff \in \tot$,

    \item[(\con--$\wedge,\vee$)] $\alpha, \beta \in \con \implies \alpha \vee \beta \in \con$ and $\alpha \wedge \beta \in \con$,

    \item[(\tot--$\wedge,\vee$)] $\alpha, \beta \in \tot \implies \alpha \vee \beta \in \tot$ and $\alpha \wedge \beta \in \tot$,

    \item[(\con--$\bigsqcup\!{}^{\upset}$)] $A\subseteq \con$ and $A$ is $\sqsubseteq$-directed $\implies \dirsqcup A \in \con$,

    \item[(\con--\tot)] $\alpha \in \con, \beta \in \tot$ and $(\alpha_+ = \beta_+$ or $\alpha_- = \beta_-) \notag \implies \alpha \sqsubseteq \beta$.
\end{axioms}
Algebraically speaking, the 3rd--6th axioms say that $(\con; \wedge,\vee, \dtt, \dff)$ and $(\tot; \wedge,\vee, \dtt, \dff)$ are (bounded) distributive lattices and that $(\con; \sqsubseteq)$ is a DCPO. Directed suprema are computed pointwise, i.e.\ for a $\sqsubseteq$-directed $A\subseteq \con$, $\dirsqcup A = (\dirvee\ \{ \alpha_+ : \alpha \in A\},\ \dirvee\ \{ \alpha_- : \alpha \in A\})$.

A pair of frame homomorphisms $h = (h_+\colon L_+\to M_+,h_-\colon L_-\to M_-)$ is a \emph{d-frame homomorphism} $h\colon \p L\to \p M$ if, for all $\alpha \in \con^\p L$, $h(\alpha) = (h_+(\alpha_+), h_-(\alpha_-)) \in \con^\p M$ and, for all $\alpha \in \tot^\p L$, $h(\alpha) \in \tot^\p M$.

Every bispace $X = (X;\, \tau_+, \tau_-)$ gives rise to a d-frame
$\Omega^d(X) = (\tau_+, \tau_-; \con^X, \tot^X)$
where $(U_+, U_-) \in \con^X$ iff $U_+\cap U_- = \emptyset$ and $(U_+, U_-) \in \tot^X$ iff $U_+ \cup U_- = X$. Similarly, every bicontinuous map $f\colon X\to Y$ gives rise to a d-frame homomorphism $\Omega^d(f) = (\Omega(f_+), \Omega(f_-))\colon \Omega^d(Y) \to \Omega^d(X)$.

\medskip

The (\con-\tot) axiom, while essential in the theory of d-frames, is harder to guarantee in constructions. We therefore introduce the auxiliary notion of a \emph{pre-d-frame} where all but the (\con-\tot) axiom of d-frames are required to hold.

\begin{remark}
    Often, when we quantify over elements or sets that appear in both plus and minus forms we will use the symbol ``$\pm$'' to mean both of them. For example, ``$A_\pm$ has property X'' means ``$A_+$ and $A_-$ have property X'', or ``there exist elements $x_\pm \in L_\pm$'' means ``there exist elements $x_+\in L_+$ and $x_-\in L_-$''.

    Also, because of the symmetrical nature of d-frames, many proofs consist of two identical arguments, one for the plus and and one for the minus side. Instead, we give only one of the variants without even mentioning the other.
\end{remark}

\section{Presentations}\label{s:presentations}

\subsection{Presentation of frames}\label{s:frame-pres}

Frames, like other algebraic structures, may be presented in terms of generators and relations $\left< G | R \right>$. The resulting frame \Fr{ G | R } is obtained as the quotient $\Fr G\!/_{\sim_R}$. Here, \Fr G represents the term algebra generated by the set of generators $G$ which, because of the frame distributivity law, consists of terms of the form: $\bigvee_i ( \wedge_{j=1}^{n_i}\ g_{i,j})$. The congruence $\sim_R$ is generated from a relation $R\subseteq \Fr G\ttimes\Fr G$ where each element of $R$ is thought of as an equation:
\begin{align}
    \bigvee_i (\wedge_{j=1}^{n_i}\ g_{i,j}) = \bigvee_{i'} ( \wedge_{j'=1}^{n'_{i'}}\ g'_{i',j'}). \label{e:R-equ}
\end{align}
However, the structure of $\Fr G\!/_{\sim_R}$ is not transparent at all. Its elements are equivalence classes of infinitary terms quotiented by $R$, which itself consists of ``infinitary'' equations.
This is addressed in the \p C-ideal presentation of frames. We assume that our generators form a meet-semilattice\footnote{We always assume that meet-semilattices are closed under \emph{all} finite meets, i.e. they also contain the top element.} $B$ representing the terms $\wedge_{j=1}^{n}\ g_{j}$. Moreover, we can restrict to equations in which the right-hand side consists of a single finite meet of generators, i.e.\ an element of $B$. In the terminology of \p C-ideals, we have a set of cover relations \p C where a \emph{cover relation} is any pair $U\dashv a$ such that $a\in B$ and $U\subseteq \downset a$ (to represent the equation $\bigvee U= a$). If, moreover, \p C satisfies the stability condition
\[ U \dashv a \in \p C,\ b\leq a \implies \{ u\wedge b : u\in U \} \dashv b \in \p C \tag{\p C-st.}\]
then we call \dBC[] a \emph{frame presentation}.

% \medskip

% A huge benefit of presenting frames this way is that we clearly separate the process into two levels of abstraction. In the first level we have \emph{syntax}, i.e.\ the presentations by \dBC[]'s, and the second is \emph{semantics}, i.e.\ frames. This is in contrast with $\Fr{G|R}$ where we mix syntax and semantics because $R$ is a subset of the products of two frames $\Fr G\!\ttimes \Fr G$.

\medskip

The frame presented by \dBC[] has an explicit description as the frame of all $\p C$-ideals, denoted by \CIdl[], where $I\subseteq B$ is a \emph{\p C-ideal} if it is a downset and
\[ U \dashv a \in \p C,\ U \subseteq I \implies a \in I \]

Computing with \p C-ideals is straightforward. The join of a set $\{I_i\}_i$ of \p C-ideal is computed as $\CIdlG[]{\bigcup_i I_i}$ where, for an $M\subseteq B$, $\CIdlG[]{M}$ is the smallest \p C-ideal containing $M$. The meets of \p C-ideals are just intersections: $\bigwedge_i I_i = \bigcap_i I_i$,~\cite[Proposition II.2.11]{johnstone82}.

There is a map translating syntactic terms to their semantic interpretation as \p C-ideals with the following universal property:

\begin{lemma}[Universality]\label{l:frm-univ}
    Let \dBC[] be a presentation of a frame. Then the map $\sem\colon B \to \CIdl[]$
    defined as $b \mapsto \CIdlG[]{\{b\}}$ is a meet-semilattice homomorphism that transforms covers into joins, i.e. $\bigvee \{ \sem[u] : u\in U\} = \sem[a]$ for every $U\dashv a\in \p C$.

    Moreover, $\sem$ is universal among all such maps. That is, if $f\colon B\to L$ is a meet-semilattice homomorphisms that transform covers in \p C into joins, where $L$ is a frame, then there exists a unique frame homomorphism $\overline f\colon \CIdl[]\to L$ such that $f = \overline f \circ \sem$.
\end{lemma}

\begin{remark}
    There are numerous ways of presenting frames, \mbox{e.g.} \cite{vickers89}, \cite{johnstone82}, \cite{ballpultr2014extending} or \cite{picadopultr2011frames}. We picked this one because it suits us better later on for the coproduct of d-frames. For the actual definition of presentation of d-frames it should not really matter as long as we have a universality property similar to the one in Lemma~\ref{l:frm-univ}.
\end{remark}

\subsection{Presentations of pre-d-frames}

In this section we show that we can extend the classical theory to also present a (pre-)d-frame \dL{}. Let us assume that $L_\pm = \CIdl$, for some frame presentations \dBC{}, as in the previous section. We also have the translations $\sem_\pm\colon B_\pm\to L_\pm$ from syntax to semantics according to Lemma~\ref{l:frm-univ}.

Any consistency relation~$\con$ on $L_+\ttimes L_-$ can be specified via the generators: Let $\alpha \in \con$. Since the sets $\sem[B_\pm]_\pm = \{ \sem[b]_\pm : b\in B_\pm \}$ generate the frames $L_\pm$,
\[ \alpha = (\bigvee_{i\in I_+} b^i_+,\  \bigvee_{i\in I_-} b^{i}_-) \quad\text{for some } \{b^i_+\}_{i} \subseteq \sem[B_+]_+ \text{ and } \{b^{i}_-\}_i \subseteq \sem[B_-]_- \]
and, because \con{} is downwards closed in the information order, \con{} must contain all the pairs $(b^{i}_+,\ b^{i'}_-)$, for $(i, i')\in I_+\ttimes I_-$. Moreover, the converse is also true:
\begin{lemma}
    $(\bigvee_{i\in I_+} b^i_+,\  \bigvee_{i\in I_-} b^{i}_-) \in \con\qtq{iff}(b^{i}_+,\ b^{i'}_-)\in \con$, for all $(i, i')\in I_+\ttimes I_-$
\end{lemma}
\begin{proof}
    Only the right-to-left implication remains to be proved. Assume that $(b^{i}_+,\ b^{i'}_-)\in \con$, for all $(i,i')\in I_+\ttimes I_-$. Since \con{} is $\wedge$-closed, for an $i\in I_+$ and a finite $F_-\finsubseteq I_-$, $(b^i_+,\, \bigvee_{i \in F_-} b^i_-)\in \con$. Similarly, since \con{} is $\vee$-closed, for finite $F_-\finsubseteq I_-$ and $F_+\finsubseteq I_+$, $(\bigvee_{i \in F_+} b^i_+,\, \bigvee_{i \in F_-} b^{i}_-)\in \con$. Notice that the set $M = \{ (\bigvee_{i \in F_+} b^i_+,\, \bigvee_{i \in F_-} b^{i}_-) : F_+\finsubseteq I_+ \text{ and } F_-\finsubseteq I_- \}$ is directed and $\dirsqcup M = (\bigvee_{i\in I_+} b^i_+,\  \bigvee_{i'\in I_-} b^{i'}_-)$. Moreover, because $M\subseteq \con$ and \con{} is closed under directed suprema, $\dirsqcup M\in \con$.
\end{proof}
This means that we can specify \con{} by a subset $\conO\subseteq B_+\ttimes B_-$ such that $\con = \CON[{\sem[\conO]}]$ where \CON[{\sem[\conO]}] is \emph{the smallest consistency relation} containing $\sem[\conO] = \{ (\sem[\alpha_+]_+, \sem[\alpha_-]_-) : \alpha\in \conO\}$.\footnote{Formally, for an $R\subseteq L_+\ttimes L_-$, \[ \CON[R] = \bigcap \{ R' \subseteq L_+\ttimes L_- ~|~ R \subseteq R', R' \text{ is $\downset$-closed, closed under $\wedge,\vee,\dirsqcup$, and $\dff,\dtt\in R'$} \}.\]}

In general, we cannot hope to do the same for \tot{}, i.e.\ find a $\totO \subseteq B_+\ttimes B_-$ such that $\tot = \TOT[{\sem[\totO]}]$ where \TOT[{\sem[\totO]}] is \emph{the smallest totality relation} containing $\sem[\totO]$. We would have to specify \tot{} by a subset of $\p P(B_+)\ttimes \p P(B_-)$. However, the special kind of presentations, when $\totO \subseteq B_+\ttimes B_-$, turns out to be sufficient for our purposes.

\begin{definition}
A tuple \dPres{} is a \emph{presentation of a pre-d-frame} if
\begin{axioms}[d-Pres]
    \item $(B_+, \p C_+)$ and $(B_-, \p C_-)$ are presentations of frames,
    \item $\conO\subseteq B_+\ttimes B_-$ and $\totO\subseteq B_+\ttimes B_-$.
\end{axioms}
The resulting pre-d-frame is obtained in two steps. First, we generate the frames of \p C-ideals \CIdl{} and then we generate the consistency and totality relations from the embedded relations $\sem[\conO], \sem[\totO] \subseteq \CIdl[+]\!\ttimes \CIdl[-]$. We obtain the following pre-d-frame:
\[  (\CIdl[+], \CIdl[-];\ \CON[{\sem[\conO]}], \TOT[{\sem[\totO]}])  \label{e:gen-d-frm}\tag{gen.} \]
\end{definition}

Similarly to its frame counterpart, $\sem\!=\!(\sem_+, \sem_-)$ has the following universal property.

\begin{lemma}[Universality]\label{l:dfrm-univ}
    Let \dPres{} be a presentation of a pre-d-frame. Then,
    $$ \sem\colon \dB \to (\CIdl[+],\, \CIdl[-];\; \CON[{\sem[\conO]}], \TOT[{\sem[\totO]}]),$$
    is presentation preserving, i.e.\ its components are meet-semilattice homomorphisms that transform covers from $\p C_\pm$ into joins and together they preserve $\conO$ and~$\totO$.

    Also, if \p M is a pre-d-frame and $f = (f_+, f_-)\colon \dB \to \p M$ is a presentation-preserving pair of maps, then there is a unique d-frame homomorphism
    \[ \overline f\colon (\CIdl[+],\, \CIdl[-];\; \CON[{\sem[\conO]}], \TOT[{\sem[\totO]}]) \to \p M\]
    such that $f = \overline f \circ \sem$. Moreover, the components of $\overline f$ are the unique frame homomorphisms that are guaranteed to exist by Lemma~\ref{l:frm-univ}.
\end{lemma}

\section{Generating d-frames}\label{s:d-frames}

So far we made no attempt in making sure that the axiom (\con-\tot) is satisfied in the generated pre-d-frame. Let us fix a presentation \dPres{} for the rest of this section and, because both frame components stay intact after we generate them, let us denote them by $L_\pm \DEFEQ \CIdl$. Also, for brevity, we will identify $B_\pm$ with $\sem[B_\pm]_\pm\subseteq L_\pm$ and, also, \conO{} and $\totO$ with $\sem[\conO]$ and $\sem[\totO]\subseteq L_+\ttimes L_-$, respectively.

The question for this section is: Under which conditions for \dPres{} is the generated pre-d-frame
$$ \dLGen $$
a d-frame? We solve this problem (partially) by showing that the following conditions are sufficient (though not necessarily minimal):
\begin{enumerate}
    \item ($\downset\conWV$-ind$_\pm$), from Section~\ref{s:one-step}, which will ensure that the structure of \CON{} is ``sufficiently simple'', and
    \item ($\lambda^4_\pm$-\con-\tot), from Section~\ref{s:chasing-contot}, which is just a simple instance of (\con-\tot).
\end{enumerate}

% TODO more about the syntax/semantics story (maybe later)?
% both conditions still require us to understand the generated frames
% we cannot prove (\con-\tot) just from syntax

\subsection{The structure of \CON{} and \TOT}

Before we get to the two conditions, we show that the relations \CON{} and \TOT{} can be generated more explicitly. As in the HSP theorem from universal algebra, we can close \conO{} and \totO{} under the operations they should be closed under (e.g.\ $\wedge$, $\vee$, etc.) and, if we proceed in a certain order, we do not have to repeat any of the steps.

Let $R\subseteq L_+\ttimes L_-$ be a any relation. We say that $R$ is $\wedge$-closed (resp.\ $\vee$-closed), if for every $\alpha, \beta\in R$, $\alpha\wedge\beta\in R$ (resp.\ $\alpha\vee\beta\in R$). By $\downset R$ denote the downwards closure of $R$ in the $\sqsubseteq$-ordering, i.e.\ the relation $\{ \alpha \in L_+\ttimes L_- ~|~ \exists \beta \in R.\ \alpha \sqsubseteq \beta \}$ and define $\upset R$ similarly.

Finally, define $\p D(R) \DEFEQ \{\, \dirsqcup A ~|~ A \dirsubseteq R \}$.\footnote{$A\dirsubseteq R$ means that $A$ is a directed subset of $R$ in the $\sqsubseteq$-order.} Note that $\p D(R)$ is only a \emph{``one-step''} closure under joins of directed subsets in $\sqsubseteq$-order. $\p D(R)$ might still contain directed subsets which do not have suprema in $\p D(R)$. To close $R$ under all directed suprema, one would have to iterate this process. However, as we will see later, there are natural conditions under which only one application is enough.

\begin{lemma}\label{l:closure-props}
    Let $L_+, L_-$ be two frames and let $R\subseteq L_+\ttimes L_-$ be a relation. Then:
    \begin{enumerate}
        \item If $R$ is $(\wedge,\vee)$-closed then $\downset R$ and $\upset R$ in $L_+\ttimes L_-$ are also $(\wedge,\vee)$-closed.
        \item If $R$ is $(\wedge,\vee)$-closed then the relation $\p D(R)$ is still $(\wedge,\vee)$-closed.
        \item If $R$ is downwards closed then the relation $\p D(R)$ is still downwards closed.
    \end{enumerate}
\end{lemma}
\begin{proof}
    For 1., let $\alpha, \beta \in \downset R$. This means that there are $\alpha', \beta'\in R$ such that $\alpha \sqsubseteq \alpha'$ and $\beta\sqsubseteq \beta'$. Observe that $(\alpha\wedge\beta)_+ = \alpha_+\wedge \beta_+ \leq \alpha'_+\wedge \beta'_+ = (\alpha'\wedge\beta')_+$ and similarly $(\alpha\wedge\beta)_- \leq (\alpha'\wedge\beta')_-$. Therefore, $\alpha\wedge\beta \sqsubseteq \alpha'\wedge\beta'\in R$ and $\alpha\wedge\beta \in \downset R$. Proving closedness $\downset R$ under $\vee$ is the same and the same reasoning also applies to $\upset R$. For~2., let $\alpha, \beta\in \p D(R)$. From the definition $\alpha = \bigsqcup\dirup_i \alpha^i$ and $\beta = \bigsqcup\dirup_j \beta^j$ for some $\alpha^i$'s and $\beta^j$'s from $R$. Let us calculate,
    \begin{align*}
        \alpha \wedge \beta
            &= (\bigvee\dirup_i \alpha^i_+\wedge\bigvee\dirup_j \beta^j_+, \bigvee\dirup_i \alpha^i_- \vee \bigvee\dirup_j \beta^j_-)\\
            &= (\bigvee\dirup_i \bigvee\dirup_j (\alpha^i_+\wedge\beta^j_+), \bigvee\dirup_i\bigvee\dirup_j (\alpha^i_- \vee \beta^j_-)) \\
            &= (\bigvee\dirup_{i,j} (\alpha^i_+\wedge\beta^j_+), \bigvee\dirup_{i,j} (\alpha^i_- \vee \beta^j_-))
    \end{align*}
    Notice that the set $\{\alpha^i\wedge\beta^j : i\in I, j\in J\}$ is directed since $\{\alpha^i \}_i$ and $\{\beta^j\}_j$ are and, moreover, $\alpha^i\wedge\beta^j\in R$ for all $i,j$ since $R$ is closed under logical meets.

    For 3., let $\beta \sqsubseteq \bigsqcup\dirup_i \alpha_i$ where $\alpha_i$'s are from $R$. Then, $\beta = \bigsqcup\dirup_i (\beta \sqcap \alpha_i)\in \p D(R)$ because the set $\{\beta \sqcap \alpha_i \}_i$ is a directed subset of $R$.
\end{proof}

Lemma~\ref{l:closure-props} shows the order in which one can generate \CON{} and \TOT{}. Set $\conWV$ to be the algebraic closure of $\conO$ under all \emph{finite} logical joins and meets in $L_+\ttimes L_-$, and define $\totWV$ correspondingly. Then we have:

\begin{corollary}
\[
\CON = \bigcup_{\iota \in \text{Ord}} \p D^\iota(\downset \conWV)
\qtq{and}
\TOT = \upset \totWV
\]
where, for an ordinal $\iota$ and a limit ordinal $\lambda$,
\[ \p \p D^0(R) \DEFEQ R,\quad D^{\iota + 1}(R) \DEFEQ \p D(\p D^\iota(R)) \qtq{and} \p D^\lambda(R) \DEFEQ \bigcup_{\iota < \lambda} \p D^\iota(R)\;.\]
\end{corollary}

\subsection{When is one step enough?}\label{s:one-step}

Proving (\con-\tot) for \CON{} and \TOT{} as it is, turned out to be too hard and, unless the authors have missed something obvious, we need to assume additional properties about the presentation. One of the reasons for the difficulty is the fact that \CON{} is computed as an iteration of $\p D(-)$. In this subsection, we focus on the question whether there are natural properties, for a relation $R\subseteq L_+\ttimes L_-$, which guarantee $\p D(\p D(R)) = \p D(R)$.

At the moment, $R$ can be any relation on the frames but for the application to presentations we would like to instantiate $R$ with $\downset \conWV$. Because of that we will assume that $R$ is downwards closed in $\sqsubseteq$-order and that it is closed under $\wedge$ and $\vee$.

We start with an important definition. Two sets $A_+ \subseteq L_+$ and $A_- \subseteq L_-$ are said to be \emph{$R$-independent} if $\forall a_+\in A_+$ and $\forall a_-\in A_-$, $(a_+, a_-)\in R$.

\begin{observation}\label{o:downset-indep}
    For every $\alpha \in R$, the sets $\p B_+(\alpha_+)$ and $\p B_-(\alpha_-)$ are $R$-independent where $\p B_\pm(\alpha_\pm) \,\DEFEQ\, \downset \alpha_\pm \cap B_\pm$.
\end{observation}

It turns out that $\p D(R)$ can reformulated by using $R$-independent sets. Let $\alpha\in \p D(R)$. From the definition, there is some directed $A\dirsubseteq R$ such that $\alpha = \dirsqcup A$. Because $\p B_\pm(-)$ are monotone and $A$ is directed, the sets $\{ \p B_+(\alpha_+) : \alpha\in A\}$ and $\{ \p B_-(\alpha_-) : \alpha\in A\}$ are both also directed (in the subset order) and so we have:
\begin{align}
    \forall A\dirsubseteq R \implies \bigcup_{\alpha\in A} \p B_+(\alpha_+) \text{ and } \bigcup_{\alpha\in A} \p B_-(\alpha_-) \text{ are $R$-independent} \label{e:dircup-indep}\tag{$\star$}
\end{align}
Moreover, because $L_\pm$ is generated by $B_\pm$ and every $x\in L_\pm$ is equal to $\bigvee \p B_\pm(x)$, we obtain that $\alpha = (\dirvee_{\alpha\in A} \alpha_+, \dirvee_{\alpha\in A} \alpha_-) = (\bigvee \p A_+, \bigvee \p A_-)$ where $\p A_\pm = \bigcup_{\alpha\in A} \p B_\pm(\alpha_\pm)$.

It might seem that $\p D(-)$ is just a special case of a more general construction:
\[ \Di(R) = \{ (\bigvee A_+,\, \bigvee A_-) ~|~ A_+\subseteq B_+,\, A_-\subseteq B_- \text{ s.t.\ } A_+ \text{ and } A_- \text{ are $R$-independent} \} \]
What we have proved in the previous paragraphs is that $\p D(R) \subseteq \Di(R)$. In fact, both closures are equivalent:

\begin{lemma}\label{l:d-equals-db}
    $\p D(R) = \Di(R)$
\end{lemma}
\begin{proof}
    Only the right-to-left inclusion remains to be proved. Let $A_+\subseteq B_+$ and $A_-\subseteq B_-$ be $R$-independent. Observe that for two finite sets $F_+ \finsubseteq A_+$ and $F_- \finsubseteq A_-$, $(\bigvee F_+, \bigvee F_-)\in R$. This is because R is $\vee$-closed and so $(\bigvee F_+, f_-)\in R$ for every $f_-\in F_-$ and, because R is $\wedge$-closed, $(\bigvee F_+, \bigvee F_-)\in R$. Clearly, the set $\p A = \{ (\bigvee F_+, \bigvee F_-) : F_+ \finsubseteq A_+ \text{ and } F_- \finsubseteq A_-\}$ is a directed subset of $R$ and $(\bigvee A_+, \bigvee A_-) = \dirsqcup \p A\in \p D(R)$.
\end{proof}

Because $\p D(R)$ is also downwards closed and closed under $\wedge$ and $\vee$ (Lemma~\ref{l:closure-props}), $\p D(\p D(R)) = \Di(\Di(R))$ and it might seem that this is already equal to $\Di(R)$. But, this is not true in general. Take, for example, $\p A_+ = \{ a_+ \}$ and $\p A_- = \{ a^1_-, a^2_- \}$ which are $\Di(R)$-independent. Each of $(a_+, a^1_-)$ and $(a_+, a^2_-)\in \Di(R)$ is witnessed by a pair of $R$-independent sets $A^1_+$ and $A^1_-$, and $A^2_+$ and $A^2_-$, respectively, such that $a_+ = \bigvee A^1_+ = \bigvee A^2_+$ and $a^1_- = \bigvee A^1_-$ and $a^2_- = \bigvee A^2_-$. However, because there is no reason to believe that $A^1_+$ and $A^2_+$ are equal, there are no obvious candidates for $R$-independent sets which would have  $(a_+, a^1_- \vee a^2_-)$ as their supremum. To overcome this problem, we assume the following condition:
\begin{axioms}
\item[($R$-ind)] For all $\forall \alpha \in \Di(R)$, $\p B_+(\alpha_+)$ and $\p B_-(\alpha_-)$ are $R$-independent.
\end{axioms}
This guarantees, for every $\alpha\in \Di(R)$, a canonical choice of $R$-independent sets, namely $A_\pm = \p B_\pm(\alpha_\pm)$.

\begin{lemma}\label{l:reduce-dirjoin}
    $\p D(\Di(R)) \subseteq \Di(R)$
\end{lemma}
\begin{proof}
    Let $A\dirsubseteq \Di(R)$. By ($R$-ind), for every $\alpha \in A$, $\p B_+(\alpha_+)$ and $\p B_-(\alpha_-)$ are $R$-independent. As in~(\ref{e:dircup-indep}), because $A$ is directed, the sets $\p A_+ \DEFEQ \dircup_{\alpha\in A} \p B_+(\alpha_+)$ and $\p A_- \DEFEQ \dircup_{\alpha\in A} \p B_+(\alpha_+)$ are $R$-independent and $\dirsqcup A = (\bigvee \p A_+, \bigvee \p A_-)$. Hence, $\dirsqcup A \in \Di(R)$.
\end{proof}

A combination of the preceding lemmas yields the desired result:

\begin{theorem}\label{t:dd-equals-d}
    Let $R\subseteq L_+\ttimes L_-$ be downwards closed, closed under logical meets and joins. If ($R$-ind) is true for $R$, then
    $\p D(\p D(R)) = \p D(R).$
\end{theorem}
\begin{proof}
$ \p D(\p D(R)) \overset{(Lemma~\ref{l:d-equals-db})}= \p D(\Di(R)) \overset{(Lemma~\ref{l:reduce-dirjoin})}\subseteq \Di(R) \overset{(Lemma~\ref{l:d-equals-db})}= \p D(R) \ \subseteq\ \p D(\p D(R))\qedhere$
\end{proof}

\begin{remark}
    Because $\p D(R) = \Di(R)$ is downwards closed, for every $\alpha\in \p D(R)$ and every $(b_+, b_-)\in \p B_+(\alpha_+)\ttimes \p B_-(\alpha_-)$, also $(b_+, b_-)\in \p D(R)$. Therefore, ($R$-ind) can be reformulated in the following more compact way:
\begin{axioms}
    \item[($R$-ind)] $(B_+\ttimes B_-) \cap \p D(R) \subseteq R$
\end{axioms}
\end{remark}

\subsection{Chasing down (\con-\tot)}\label{s:chasing-contot}
Finally, we can focus on the original (\con-\tot) axiom for \dLGen. We split it into two parts:
\begin{axioms}
    \item[($\lambda^0_+$-\con-\tot)] $\alpha\in \CON$, $\beta\in \TOT$ and $\alpha_+ = \beta_+ \implies \alpha_- \leq \beta_-$
    \item[($\lambda^0_-$-\con-\tot)] $\alpha\in \CON$, $\beta\in \TOT$ and $\alpha_- = \beta_- \implies \alpha_+ \leq \beta_+$
\end{axioms}
If we assume ($R$-ind) about $\downset \conWV$, then the conditions of Theorem~\ref{t:dd-equals-d} hold for $R = \downset \conWV$ and we can rewrite ($\lambda^0_\pm$-\con-\tot) into the following more explicit form
\begin{axioms}
\item[($\lambda^0_+$-\con-\tot)] $\alpha\in \sCON$, $\beta\in \upset \totWV$ and $\alpha_+ = \beta_+ \implies \alpha_- \leq \beta_-$
\item[($\lambda^0_-$-\con-\tot)] $\alpha\in \sCON$, $\beta\in \upset \totWV$ and $\alpha_- = \beta_- \implies \alpha_+ \leq \beta_+$
\end{axioms}

Our aim now is to restrict $\alpha$ and $\beta$ to smaller and smaller sets. First, we restate the axioms such that the $\beta$'s come from $\totWV$ and then from $\totW$ (resp.\ $\totV$). Then, we do the same with $\alpha$ until we obtain a version of the (\con-\tot) axiom stated purely in terms of formulas involving only elements from $\conWbV$ (resp.\ $\conVbW$) and $\totW$ (resp.\ $\totV$). The individual stages are depicted in the diagram below (the $\lambda$ superscripts in the axiom name correspond to the stages as shown in the diagram):

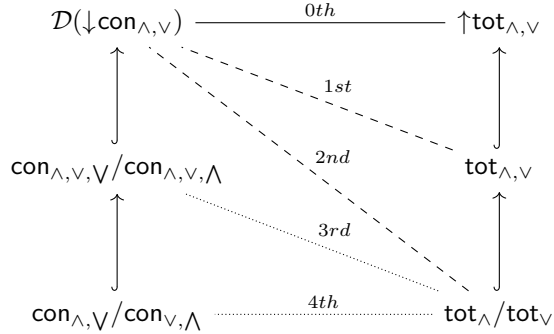
\begin{diagram}
    \sCON & & \upset \totWV \ar[rel,solid]{ll}{0th} \\
    \conWVbV/\conWVbW\ar[hook]{u}  & & \totWV\ar[hook]{u}\ar[rel]{llu}{1st} \\
    \conWbV/\conVbW\ar[hook]{u}  & & \totW/\totV\ar[hook]{u}\ar[rel]{lluu}{2nd}\ar[rel,densely dotted]{llu}{3rd}\ar[rel,densely dotted]{ll}{4th}
    % \conO\ar[hook]{u}  & & \totO\ar[hook]{u} \\
\end{diagram}

In every stage we introduce a pair of axioms (named ($\lambda^i_\pm$-\con-\tot), for $i=1,\dots,4$) and show that they imply the previous axioms. Because the axioms ($\lambda^i_+$-\con-\tot) and ($\lambda^i_-$-\con-\tot) are dual to each other, we will always only prove that, say, ($\lambda^i_+$-\con-\tot) implies ($\lambda^{i-1}_+$-\con-\tot) and leave out that ($\lambda^i_-$-\con-\tot) implies ($\lambda^{i-1}_-$-\con-\tot) as it is proved dually.

\begin{remark}
    Above we use a notation similar to the one introduced earlier. The relation $\conV$ is the algebraic closure of $\conO$ under all \emph{finite} logical joins ($\vee$) in $L_+\ttimes L_-$, and $\conW$, $\totV$, $\totW$, $\totWV$ and $\conWV$ are defined correspondingly. Likewise, \conWbV{} is the closure of $\conO$ under finite meets followed by the closure under all joins, both in logical order\footnote{This makes sense because, in any d-frame \dL, $\{ (\bigvee_i \alpha^{i}_+,\ \bigwedge_i \alpha^i_-) : \{ \alpha^{i} \}_{i} \subseteq \con \} \subseteq \con$. Indeed, from (\con--$\downset$), all $(\alpha^{i}_+,\ \bigwedge_i \alpha^i_-)\in \con$ and, by $\vee$ and $\dirsqcup$-closedness, $(\bigvee _i\alpha^{i}_+,\ \bigwedge_i \alpha^i_-)\in \con$. }, i.e.\
\[ \conWbV = \{ (\bigvee_i \alpha^{i}_+,\ \bigwedge_i \alpha^i_-) : \{ \alpha^{i} \}_{i} \subseteq \conW \}. \]
The other versions, such as \conVbW, \conWVbV{} and \conWVbW, are defined correspondingly.
\end{remark}

\paragraph*{1st stage.}
We intend to simplify the elements in the \tot{} relation. Consider the following axioms:
\begin{axioms}
    \item[($\lambda^1_+$-\con-\tot)] $\alpha\in \sCON,\ \higlightChange{\beta \in \totWV},\ \beta_+\leq \alpha_+ \implies \alpha_-\leq \beta_-$
    \item[($\lambda^1_-$-\con-\tot)] $\alpha\in \sCON,\ \higlightChange{\beta \in \totWV},\ \beta_-\leq \alpha_- \implies \alpha_+\leq  \beta_+$
\end{axioms}
Now, let $\alpha \in \sCON$ and let $\beta \in \TOT$ with $\alpha_+=\beta_+$. That means that there is some $\beta'\in \totWV$ such that $\beta'\sqsubseteq \beta$. We have that $\beta'_+ \leq \alpha_+$ and so we can now apply ($\lambda^1_+$-\con-\tot) and get that $\alpha_- \leq \beta'_-$ and so $\alpha_-\leq \beta'_- \leq \beta_-$. To sum up, we have proved the first part of:
\begin{lemma}
($\lambda^1_\pm$-\con-\tot) implies ($\lambda^0_\pm$-\con-\tot), and vice versa.
\end{lemma}
For the converse assume $\beta_+\leq \alpha_+$. Then the pair $(\alpha_+,\beta_-)$ belongs to $\upset \totWV$ and by ($\lambda^0_\pm$-\con-\tot) we can conclude $\alpha_-\leq\beta_-$.

\paragraph*{2nd stage.}
We can simplify the elements in \tot{} even further. Take the axioms:
\begin{axioms}
    \item[($\lambda^2_+$-\con-\tot)] $\alpha\in \sCON,\ \higlightChange{\beta \in \totW},\ \beta_+\leq \alpha_+ \implies \alpha_-\leq \beta_-$
    \item[($\lambda^2_-$-\con-\tot)] $\alpha\in \sCON,\ \higlightChange{\beta \in \totV},\ \beta_-\leq \alpha_- \implies \alpha_+\leq  \beta_+$
\end{axioms}
Let $\alpha \in \sCON$ and let $\beta \in \totWV$ with $\alpha_+\leq\beta_+$. We can decompose $\beta$ such that $\beta = \bigvee_{k=1}^n \beta^k$ where $\beta^k \in \totW$, for every $k=1,\dots,n$. Then, for every $k$, we have that $\beta^k_+ \leq \beta_+ \leq \alpha$ and so $\alpha_-\leq \beta^k_-$. Because $\alpha_-\leq \beta^k_-$ for every $k$, also $\alpha_- \leq \beta_- = \bigwedge_{k=1}^n \beta^k_-$.

\begin{lemma}
($\lambda^2_\pm$-\con-\tot) implies ($\lambda^1_\pm$-\con-\tot), and vice versa.
\end{lemma}
Here the converse direction is trivial.

\paragraph*{3rd stage.} Now we focus on the complexity of elements~$\alpha$ from \con{}. To eliminate $\p D(-)$ consider the following auxiliary axioms:
\begin{axioms}
\item[($\alpha_+$-\con-\tot)] $\higlightChange{\{ (x^k, y) \}_k \subseteq \downset \conWV},\ \beta\in \totW,\ \beta_+ \leq \bigvee_k x^k \implies y\leq \beta_-$

\item[($\alpha_-$-\con-\tot)] $\higlightChange{\{ (x, y^k) \}_k \subseteq \downset \conWV},\ \beta\in \totV,\ \beta_- \leq \bigvee_k y^k \implies x\leq \beta_+$
\end{axioms}
Let $\alpha\in \sCON$. By Lemma~\ref{l:d-equals-db}, this means that there exist $A_\pm\subseteq B_\pm$ which are ($\downset\conWV$)-independent and such that $\alpha = (\bigvee A_+, \bigvee A_-)$. Let us fix a $b_- \in A_-$. The ($\downset\conWV$)-independence of $A_+$ and $A_-$ means that $A_+ \ttimes \{b\} \subseteq \downset \conWV$. Because also $\beta_+ \leq \alpha_+ = \bigvee A_+$, we can apply ($\alpha_+$-\con-\tot) and obtain that $b_- \leq \beta_-$. Since $b_-\in A_-$ was chosen arbitrarily, $\alpha_- = \bigvee A_- \leq \beta_-$. We have proved that ($\alpha_+$-\con-\tot) implies ($\lambda^2_+$-\con-\tot).

Finally, we can get rid of the downwards closure of \conWV{}. Consider the following axioms:

\begin{axioms}
    \item[($\lambda^3_+$-\con-\tot)] $\higlightChange{\alpha\in \conWVbV},\ \beta \in \totW$, $\beta_+ \leq \alpha_+ \implies \alpha_- \leq \beta_-$

    \item[($\lambda^3_-$-\con-\tot)] $\higlightChange{\alpha\in \conWVbW},\ \beta \in \totV$, $\beta_- \leq \alpha_- \implies \alpha_+ \leq \beta_+$
\end{axioms}

Let $\{ (x^k, y) \}_k \subseteq \downset \conWV$ be such that $\beta_+ \leq \bigvee_k x^k$. For every $k$, there exists an $\alpha^k\in \conWV$ such that $(x^k, y)\sqsubseteq \alpha^k$. Clearly, $\beta_+ \leq \bigvee_k x^k \leq \bigvee_k \alpha^k_+$, and $\alpha = (\bigvee_k \alpha^k_+, \bigwedge_k \alpha^k_-)\in \conWVbV$. We can apply ($\lambda^3_+$-\con-\tot) and obtain that $y \leq \alpha_- \leq \beta_-$. Together with the previous result we have that:

\begin{lemma}
($\lambda^3_\pm$-\con-\tot) implies ($\lambda^2_\pm$-\con-\tot).
\end{lemma}

\paragraph*{4th stage.} The final simplification is similar to the 2nd stage but this time acts on the $\con$ side:
\begin{axioms}
    \item[($\lambda^4_+$-\con-\tot)] $\higlightChange{\alpha\in \conWbV},\ \beta \in \totW$, $\beta_+ \leq \alpha_+ \implies \alpha_- \leq \beta_-$

    \item[($\lambda^4_-$-\con-\tot)] $\higlightChange{\alpha\in \conVbW},\ \beta \in \totV$, $\beta_- \leq \alpha_- \implies \alpha_+ \leq \beta_+$
\end{axioms}
Distributivity of $\wedge$ and $\vee$ gives us that
\begin{align*}
    \conWVbV = \conWbV \qtq{and} \conWVbW = \conVbW
\end{align*}
from which we can conclude:
\begin{lemma}
    ($\lambda^4_\pm$-\con-\tot) implies ($\lambda^3_\pm$-\con-\tot), and vice versa.
\end{lemma}

Furthermore, ($\lambda^4_\pm$-\con-\tot) and ($\lambda^1_\pm$-\con-\tot) are equivalent because
\begin{align}
    \conWbV \subseteq \sCON \qtq{and} \conVbW \subseteq \sCON. \label{e:con-subsets}
\end{align}
To prove these inclusions, let $\alpha = (\bigvee_k \alpha^k_+, \bigwedge_k \alpha^k_-)$ where $\{ \alpha^k \}_{k\in K} \subseteq \conW$. Then, for every $k\in K$, $(\alpha^k_+, \alpha_-)\sqsubseteq \alpha^k$ and so $(\alpha^k_+, \alpha_-) \in \downset \conW \subseteq \downset \conWV$. Because $\downset \conWV$ is $\vee$-closed, $\{ (\bigvee_{k\in F}\, \alpha^k_+, \alpha_-) : F\finsubseteq K \}$ is directed in $\downset \conWV$ and so $\alpha \in \sCON$.

\medskip
We can apply similar techniques to simplify ($R$-ind):

\begin{lemma}\label{l:indep-split}
    ($\downset\conWV$-ind) is equivalent to having the following two conditions
\begin{axioms}
    \item[($\downset\conWV$-ind$_+$)] $(B_+\ttimes B_-) \cap \downset \conWbV \subseteq \downset \conWV$
    \item[($\downset\conWV$-ind$_-$)] $(B_+\ttimes B_-) \cap \downset \conVbW \subseteq \downset \conWV$
\end{axioms}
\end{lemma}

We sum up all the previous results into this theorem:
\begin{theorem}\label{t:con-tot}
    If ($\lambda^4_\pm$-\con-\tot) and ($\downset\conWV$-ind$_\pm$) hold for a pre-d-frame presentation, then the generated pre-d-frame satisfies (\con-\tot).
\end{theorem}

\begin{remark}
    It is not possible to check if (\con-\tot) holds in the generated pre-d-frame just by looking at its syntactic presentation. However, our sufficient conditions are much simpler than the formulas involving infinitary applications of $\p D(-)$. Nevertheless we still need to understand the structure of the generated frame components.
    % can then be checked much earlier in the generation process, both conditions still require us to understand the generated frames, we cannot prove (\con-\tot) just from syntax
\end{remark}

\subsection{A special case}
In our applications even stronger and simpler conditions hold for the presentations. Namely, consider the following \emph{``micro version''} of (\con-\tot):
\begin{axioms}
    \item[($\mu_+$-\con-\tot)] $\alpha \in \conV$,   $\beta\in \totW$. $\beta_+ \leq \alpha_+ \implies \alpha_- \leq \beta_-$
    \item[($\mu_-$-\con-\tot)] $\alpha \in \conW$, $\beta\in \totV$.   $\beta_- \leq \alpha_- \implies \alpha_+ \leq \beta_+$
\end{axioms}
and the following (more powerful) version of conditions ($\downset\conVW$-ind$_\pm$):
\begin{axioms}
    \item[(Indep$_+$)] $(L_+\ttimes B_-) \cap \downset \conWbV \subseteq \downset \conV$
    \item[(Indep$_-$)] $(B_+\ttimes L_-) \cap \downset \conVbW \subseteq \downset \conW$
\end{axioms}

\begin{proposition}\label{p:simple-axioms}
    If ($\mu_\pm$-\con-\tot) and (Indep$_\pm$) hold for a pre-d-frame presentation, then the generated pre-d-frame satisfies (\con-\tot).
\end{proposition}

\section{Application: Coproducts}\label{s:coproducts}

\subsection{Coproducts of frames}\label{s:coprod-frm}
For a nice presentation of the coproducts of frames, we refer the reader to the book \emph{``Frames and Locales''}~\cite{picadopultr2011frames}. Here we only outline basic facts about the construction. Let $\{L^i\}_{i\in \p I}$ be a family of frames. The coproduct of $\{ L^i \}_i$ in the category of meet-semilattices is $\prod'_i L^i$ which is the subset of $\prod_i L^i$ consisting of those elements with all but finitely many coordinates equal to 1. Then, the coproduct of $\{L^i\}_i$ in the category of frames $\bigoplus_i L^i$ can be presented as the frame of $\p C$-ideals of $(\prod'_i L^i, \p C)$ with the set of coverings $\p C$ of the form:
\[ \{ a^k *_j u : k\in K \} \dashv (\bigvee_{k\in K} a^k) *_j u \]
where, for an $a\in L^j$ and $u\in \prod'_i L^i$, $a *_j u$ is the element of $\prod'_i L^i$ such that $(a *_j u)_j = a$ and $(a *_j u)_i = u_i$ for $i\not=j$. Recall also that the smallest element of $\bigoplus_i L^i$ is the \p C-ideal $\mathbf n = \{ u \in \prod'_i L^i ~|~ u_i = 0 \text{ for some } i\}$.

The inclusion maps are the frame homomorphisms $\iota^j\colon L^j \to \bigoplus_i L^i$, $x \mapsto \downset (x *_j \ones) \cup \mathbf n$, where $(\ones)_i = 1$ for all $i\in \p I$. We can factor $\iota^j$ into a composition of two meet-semilattice homomorphisms $\sem\circ \kappa^j$ where
\begin{align*}
\kappa^j\colon L^j &\to \prod\nolimits_i' L^i
&\qtq{and}
&&\sem\colon \prod\nolimits'_i L^i &\to \bigoplus_i L^i\\
x&\mapsto x *_j \ones
&&& u &\mapsto \downset u \cup \mathbf n
\end{align*}
Here $\kappa^j$ is the universal map for the semilattice coproduct $\prod\nolimits_i' L^i$ and $\sem$ is the inclusion $B \to \CIdl[]$ as in Lemma~\ref{l:frm-univ}, $B = \prod'_i L^i$, and $\downset u \cup \mathbf n$ is the smallest \p C-ideal containing~$u$.

\subsection{Coproducts of d-frames}\label{s:dfrm-coproducts}
Let $\{ \p L^i = (L_+^i, L_-^i; \con^i, \tot^i) \}_{i\in \p I}$ be a family of d-frames. We will define the coproduct of $\{ \p L^i\}_i$ by a~free \mbox{d-frame} construction. First, we compute the frame components of the coproduct of $\{ \p L^i \}_i$ as the coproducts of the frame components of the d-frames $\{ \p L^i\}_i$. Set $B_+ = \prod'_i L_+^i$ and $B_- = \prod'_i L_-^i$ and $\p C_+$ and $\p C_-$ independently as in Subsection~\ref{s:coprod-frm} (for $B_+$ and $B_-$, respectively). Namely, for every $j\in \p I$, we have a frame homomorphism
\[
\begin{tikzcd}
    \iota^j_\pm\colon L^j_\pm\ar{r}{\kappa^j_\pm} & \prod'_i L^i_\pm \ar{r}{\sem_\pm} & \bigoplus_i L^i_\pm
\end{tikzcd}
\]

In order for $\iota^j = (\iota^j_+, \iota^j_-)$ to be a \emph{d-frame} embedding into a coproduct, for every $(a,b)\in \con^j$ (resp.\ $\tot^j$), it has to be the case that $(\iota^j_+(a),\, \iota^j_-(b)) \in \CON$ (resp.\ $\TOT$). Also, the universal property of coproducts guarantees that for any d-frame cone $\{ \p L^i \to \p M \}_i$ there is a mediating d-frame homomorphism $\bigoplus_i \p L^i\to \p M$. This means that the relations we generate $\CON$ and $\TOT$ from should not contain anything more. Therefore, define $\conO, \totO\subseteq B_+\ttimes B_-$ by
\begin{align*}
    (a *_j \ones,\, b *_j \ones) \in \conO &\qtq{iff} (a,b)\in \con^j \\
    (a *_j \ones,\, b *_j \ones) \in \totO &\qtq{iff} (a,b)\in \tot^j
\end{align*}
and by $\bigoplus_i \p L^i$ denote the resulting pre-d-frame $(\OLxOL;\, \CON, \TOT)$.

\begin{notation}
For every $a\in L^i$ and $u\in B$, denote $a \oplus_i u = \sem[a *_i u]_\pm = \downset (a *_i u)\cup \mathbf n_\pm$. In particular, $a \oplus_i \ones = \downset (a *_i \ones) \cup \mathbf n_\pm$. As before, we identify $B_\pm$ with $\sem[B_\pm]_\pm\subseteq \bigoplus_i L^i_\pm$, $\conO$ with $\sem[\conO] \subseteq \OLxOL$, and $\totO$ with $\sem[\totO]$.
\end{notation}

To simplify our work by making sure that we can deal with indexes coherently, we prove the following lemma about normal forms of elements from \conV, \conW, \totV{} and \totW{}:

\begin{lemma}\label{p:contot-canon}
    Let $\alpha\in \conW/\totW$. Then, it is of the form $(\bigwedge_i \alpha^i_+, \bigvee_i \alpha^i_-)$ such that
\begin{enumerate}
    \item for every $i\in \p I$: $\alpha^i = (a_+\oplus_i \ones, a_-\oplus_i \ones)$ for some $(a_+, a_-)\in \con^i$ (resp. $\tot^i$), and
    \item there exists a \emph{finite} $I(\alpha) \finsubseteq \p I$ s.t.\ $i\in I(\alpha)$ iff $\alpha^i \not= \dtt$
\end{enumerate}

Similarly, every $\alpha\in \conV$ (resp.\ $\totW$) is of the form $(\bigvee_i \alpha^i_+, \bigwedge_i \alpha^i_-)$ where $\alpha^i\in \conO$ (resp. $\totO$) and $I(\alpha)$ denotes the finite set of indexes for which $\alpha^i\not= \dff$.
\end{lemma}

Notice that 1.\ and 2.\ make sense together. Anytime $\alpha^i = \dtt$ we have that $\dtt = (\downset \ones \cup \np, \nm) = (1 \oplus_i \ones, 0 \oplus_i \ones)$ and $(1, 0) \in \con^i/\tot^i$. The case for $\alpha^i = \dff$ is similar.

\subsection{Strips, rectangles and crosses}

Before we get into proving that $\bigoplus_i \p L^i$ satisfies (\con-\tot) we look into the structure of \conV, \conW, \totV{} and \totW. It turns out that there is a nice geometrical intuition that we can employ.

First, for an $a\in L^i_\pm$, we call $a \oplus_i \ones$ \emph{an $i$-strip}\footnote{We sometimes omit the index and call $i$-strips just strips whenever it does not lead to a confusion.}. Then, anytime $(a,b) \in \con^i$, we can think of the corresponding pair $(a \oplus_i \ones,\, b \oplus_i \ones)\in \conO$ as of a pair of \emph{``disjoint''} $i$-strips and, similarly, $(c,d) \in \tot^i$ gives a pair of strips that are \emph{``covering the whole space''}, i.e. $(c \oplus_i \ones,\, d \oplus_i \ones)\in \totO$. This terminology is motivated by the case when $\p I = \{1,2\}$. Both cases are displayed in the picture below for $\p L^1 \oplus \p L^2$:

\tikzset{
    % plus
    pls/.style = { dashed, pattern=pluses , pattern color=blue, draw opacity=0.3},
    plscol/.style = { color = blue},
    % minus
    mns/.style = { dashed, pattern=minuses, pattern color=red,  draw opacity=0.3},
    mnscol/.style = { color = red},
    % intervals
    int/.style = { very thick, <->, >=stealth },
    intname/.style = { midway, above },
}
\newcommand\drawframe{
    \draw (0,0) rectangle (3, 3);
    \draw (1.5,-0.3) node {$\p L^1$};
    \draw (-0.3,1.5) node {$\p L^2$};
}
\newcommand\mathwithshadow[1]{
    % node[pos=.5, scale=1.1, text=white] {$\boldsymbol{#1}$} % for white shadow
    node[pos=.5] {$#1$}
}
\begin{center}
\begin{tikzpicture}
    \begin{scope}
    \draw[pls] (1.8,0) rectangle (2.8, 3) \mathwithshadow{a\oplus_1 \ones};
    \draw[mns] (0.3,0) rectangle (1.3, 3) \mathwithshadow{a\oplus_1 \ones};

    \draw[int,plscol] (1.8,3.5) -- (2.8,3.5) node[intname] {a};
    \draw[int,mnscol] (0.3,3.5) -- (1.3,3.5) node[intname] {b};
    \drawframe
    \end{scope}

    \begin{scope}[shift={(5,0)}]
    \draw[pls] (1.2,0) rectangle (3, 3) \mathwithshadow{c\oplus_1 \ones};
    \draw[mns] (0,0) rectangle (1.8, 3) \mathwithshadow{d\oplus_1 \ones};

    \draw[int,plscol] (1.2,3.5) -- (3,3.5) node[intname] {c};
    \draw[int,mnscol] (0,3.6) -- (1.8,3.6) node[intname] {d};
    \drawframe
    \end{scope}
\end{tikzpicture}
\end{center}
Therefore, all elements of $\conO$ and $\totO$ are pairs of strips. It is rather a technical lemma that the set of $i$-strips in the coproduct has exactly the same structure as the d-frame $\p L^i$:

\begin{lemma}\label{p:strips}
    Let $S^i_\pm$ be the set of all $i$-strips in $\bigoplus_i L^i_\pm$. If all $L^i_\pm$'s are nontrivial\footnote{A frame is \emph{trivial} if it is isomorphic to the trivial frame $\mathbf 1 = \{ 0 = 1 \}$.} then
    \[ (S^i_+, S^i_-;\, \conO\cap (S^i_+\ttimes S^i_-),\, \totO\cap (S^i_+\ttimes S^i_-)) \cong \p L^i. \]
\end{lemma}

Moreover, finite $\wedge$-combinations of pairs of strips is something that we can imagine as a pair consisting of a rectangle and a cross. For example, let $\alpha \in \conO$ be a pair of $1$-strips and $\alpha'\in \conO$ a pair of $2$-strips. Then, as the picture below suggests, the plus coordinate of $\alpha \wedge \alpha'$ in $\p L^1 \oplus \p L^2$ is a rectangle and the minus coordinate is a cross. Notice also that the cross and rectangle are disjoint.

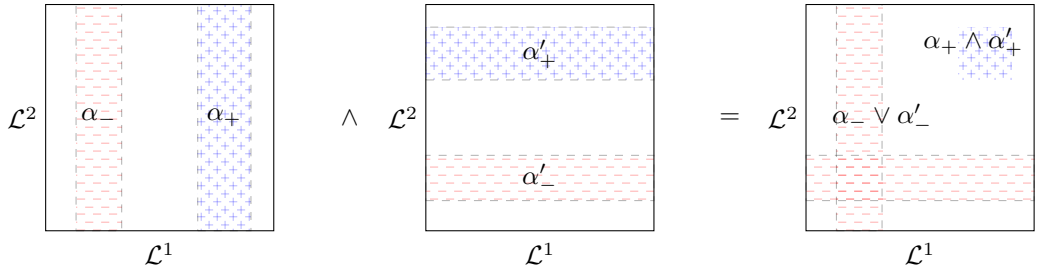
\begin{center}
\begin{tikzpicture}
    \begin{scope}
        \draw[pls] (2,0) rectangle (2.7, 3) \mathwithshadow{\alpha_+};
        \draw[mns] (0.4,0) rectangle (1, 3) \mathwithshadow{\alpha_-};
        \drawframe
    \end{scope}

    \node at (4,1.5) {$\wedge$};

    \begin{scope}[shift={(5,0)}]
        \draw[pls] (0,2) rectangle (3, 2.7) \mathwithshadow{\alpha'_+};
        \draw[mns] (0,0.4) rectangle (3, 1) \mathwithshadow{\alpha'_-};
        \drawframe
    \end{scope}

    \node at (9,1.5) {$=$};

    \begin{scope}[shift={(10,0)}]
        \begin{scope}
        \clip (0,2) rectangle (2.7, 2.7);
        \fill[pls] (2,0) rectangle (3, 3);
        \end{scope}

        % \draw[dashed] (0,2) rectangle (3, 2.7);
        % \draw[dashed] (2,0) rectangle (2.7, 3);

        \draw[mns] (0,0.4) rectangle (3, 1);
        \draw[mns] (0.4,0) rectangle (1, 3);

        \node at (1.0,1.5) {$\alpha_-\vee\alpha'_-$};
        \node at (2.2,2.5) {$\alpha_+\wedge\alpha'_+$};
        \drawframe
    \end{scope}
\end{tikzpicture}
\end{center}
The picture for two pairs of strips $\beta, \beta'\in \totO$ is similar but this time the cross and rectangle of $\beta\wedge \beta'$ cover the whole space.
% Similarly, for two pairs of strips $\beta, \beta'\in \totO$, we have a cross and rectangle covering the whole space:
% \newcommand\recTPa{(1.2,0) rectangle (3, 3)}
% \newcommand\recTMa{(0,0) rectangle (1.8, 3)}
% \newcommand\recTPb{(0,1.2) rectangle (3, 3)}
% \newcommand\recTMb{(0,0)   rectangle (3, 1.8)}
% \begin{center}
% \begin{tikzpicture}
%     \begin{scope}
%         \draw[pls] \recTPa \mathwithshadow{\beta_+};
%         \draw[mns] \recTMa \mathwithshadow{\beta_-};
%     \drawframe
%     \end{scope}
% 
%     \node at (4,1.5) {$\wedge$};
% 
%     \begin{scope}[shift={(5,0)}]
%     \draw[pls] \recTPb \mathwithshadow{\beta'_+};
%     \draw[mns] \recTMb \mathwithshadow{\beta'_-};
%     \drawframe
%     \end{scope}
% 
%     \node at (9,1.5) {$=$};
% 
%     \begin{scope}[shift={(10,0)}]
%         \begin{scope}
%         \clip \recTPa;
%         \fill[pls] \recTPb;
%         \end{scope}
% 
%         % \draw[dashed] \recTPa;
%         % \draw[dashed] \recTPb;
% 
%         \draw[mns]    \recTMa;
%         \draw[mns]    \recTMb;
% 
%         \node at (0.8,0.5) {$\beta_-\vee  \beta'_-$};
%         \node at (2.2,2.5) {$\beta_+\wedge\beta'_+$};
%         \drawframe
%     \end{scope}
% \end{tikzpicture}
% \end{center}

This geometrical intuition builds up well for these formal definitions:
$\gamma = \bigwedge_i \gamma^i$, where $\gamma^i = c^i \oplus_i \ones$ ($\forall i\in \p I$), is a \emph{rectangle} if there exists a finite $I(\gamma) \finsubseteq \p I$ such that $c^i \not= 1$ iff $i\in I(\gamma)$. Similarly, $\delta = \bigvee_i \delta^i$, where $\delta^i = d^i \oplus_i \ones$, is a \emph{cross} if for some finite $I(\delta)\finsubseteq \p I$, $d^i \not= 0$ iff $i\in I(\delta)$.

Notice that, by Lemma~\ref{p:contot-canon}, every element of \conW{} (resp.\ \totW{}) is of the form $(\bigwedge_i \alpha^i_+, \bigvee_i \alpha^i_-)$ with only finitely many nontrivial $\alpha^i$'s. In the present terminology, $\alpha$ is a pair rectangle--cross and this exactly matches the geometrical intuition we just discussed.

\begin{observation}\label{o:rectangles}
    Rectangles are exactly the elements of $B_\pm$.
\end{observation}
\begin{proof}
    Every $\gamma \in B_\pm$ is of the form $\sem[u]_\pm$ for some $u\in \prod'_i L^i_\pm$. Because $u$ has only finitely many indexes different from 1, $\sem[u]_\pm = \sem[(a^1 *_{i(1)} \ones) \wedge \dots \wedge (a^n *_{i(n)} \ones)]_\pm = \sem[a^1 *_{i(1)} \ones]_\pm \wedge \dots \wedge \sem[a^n *_{i(n)} \ones]_\pm = (a^1 \oplus_{i(1)} \ones) \wedge \dots \wedge (a^n \oplus_{i(n)} \ones)$. The reverse direction is similar.
\end{proof}

There is a nice interplay between rectangles and crosses:

\begin{lemma}\label{p:rec-cross}
    Let $\gamma = \bigwedge_i \gamma^i$ be a rectangle and let $\delta = \bigvee_i \delta^i$ be a cross such that $\gamma \leq \delta$. Then, there exists an $i\in I(\gamma)$ such that $\gamma^i \leq \delta^i$.
\end{lemma}
\begin{proof}
    Let $\gamma^i = c^i \oplus_i \ones$ and $\delta^i = d^i \oplus_i \ones$, for every $i\in \p I$. By Observation~\ref{o:rectangles}, $\gamma = \sem[u]_\pm$ for some $u\in \prod'_i L^i_\pm$ such that, for every $i \in \p I$, $(u)_i = c^i$. This means that $(u)_i \not= 1$ iff $i\in I(\gamma)$. Also, by Lemma~\ref{l:oplus-vee}, $\delta$ has a form of a finite union $\bigcup_{i\in I(\delta)} \delta^i$. If $c^i = 0$ for some $i\in I(\gamma)$, then $c^i \leq d^i$. Otherwise, $c^i \not= 0$ for all $i\in I(\gamma)$ and, since $\gamma\leq \delta$ iff $u\in \delta$, there must exist an $i\in I(\delta)$ such that $u\in \delta^i$ and then, by Lemma~\ref{p:copr-basics}.1, $(u)_i = c^i \leq d^i$. Finally, because $i\in I(\delta)$, $d^i\not= 1$ and so also $c^i\not= 1$ and $i\in I(\gamma)$.
\end{proof}

\subsection{Proof of (\con-\tot)}

In this section we prove that $\bigoplus_i \p L^i$ is a d-frame. To simplify our proofs, we can assume that all $\p L^i$'s are nontrivial thanks to the following lemma.

\begin{lemma}\label{l:trivial-case}
    If $L^i_+ = \mathbf 1$ or $L^i_- = \mathbf 1$ for some $i\in \p I$, then $\bigoplus_i \p L^i$ satisfies (\con-\tot).
\end{lemma}
\begin{proof}
    Observe that, by (\con-\tot) for $\p L^i$, if $L^i_+ = \mathbf 1$ then automatically also $L^i_- = \mathbf 1$, and vice versa. Therefore, $\bigoplus_i L^i_\pm = \{\npm\}$ and so $\bigoplus_i \p L^i$ is trivial and satisfies (\con-\tot).
\end{proof}

To show that (\con-\tot) holds for $\bigoplus_i \p L^i$ we will use Proposition~\ref{p:simple-axioms}. In order to be able to do that we need to prove that ($\mu_\pm$-\con-\tot) and (Indep$_\pm$) hold:

\begin{lemma}
    ($\mu_\pm$-\con-\tot) holds for $\bigoplus_i \p L^i$:
\end{lemma}
\begin{proof}
    Let $\alpha = \bigvee_i \alpha^i \in \conV$ and $\beta = \bigwedge_i \beta^i \in \totW$ be in canonical forms, and assume that $\beta_+ \leq \alpha_+$. From canonicity of $\alpha$ and $\beta$, know that $\alpha_+$ is a cross and $\beta_+$ is a rectangle. By Lemma~\ref{p:rec-cross}, there is an $i\in I(\beta)$ such that $\beta^i_+ \leq \alpha^i_+$. From (\con-\tot) for $\p L^i$, $\alpha^i_- \leq \beta^i_-$ and so
    $\alpha_- = \bigwedge_i \alpha^i_- \,\leq\, \alpha^i_- \,\leq\, \beta^i_- \,\leq\, \bigvee_i \beta^i_- = \beta_-.$
\end{proof}

\begin{lemma}
    (Indep$_\pm$) holds for $\bigoplus_i \p L^i$.
\end{lemma}
\begin{proof}
    Let $(x,b_-)\in (L_+\ttimes B_-)\cap \downset \conWbV$. Denote its upper bound $(\bigvee_k \alpha^k_+, \bigwedge_k \alpha^k_-)$ where, for each $k$, $\alpha^k = (\bigwedge_i \alpha^{k,i}_+, \bigvee_i \alpha^{k,i}_-)$ is a pair rectangle--cross from \conW. Because $b_-\in B_-$, it is a rectangle of the form $b_- = \bigwedge_i \gamma^i$ (Observation~\ref{o:rectangles}).  Because, for every $k$, $b_-\leq \alpha^k_-$, by Lemma~\ref{p:rec-cross}, there exists an $i(k)\in I(b_-)$ such that $\gamma^{i(k)} \leq \alpha^{k,i(k)}_-$. Fix an $i\in I(b_-)$ and set $K(i) = \{ k ~|~ i(k) = i \}$.
    By Lemma~\ref{l:trivial-case}, we can assume that all $L^i_\pm$'s are nontrivial and because $\{ \alpha^{k, i} : k\in K(i)\}$ are all pairs of $i$-strips and $\gamma^i$ is an $i$-strip, by Lemma~\ref{p:strips}, we can carry the reasoning in the rest of this paragraph in the d-frame $\p L^i$. Since $\con^{i(k)}$ is downwards closed and $\gamma^i\leq \alpha^{k,i}$ ($\forall k\in K(i)$), also $(\alpha^{k, i(k)}_+, \gamma^{i(k)}) \in \conO$ and, therefore, by $\dirsqcup$ and $\vee$-closeness of $\con^{i(k)}$, $(\bigvee_{k\in K(i)} \alpha^{k,i}_+, \gamma^i)\in \conO$.

    Finally, because $I(b_-)$ is finite
    \[ \bigvee_{i\in I(b_-)} (\bigvee_{k\in K(i)} \alpha^{k,i}_+, \gamma^i) =  (\bigvee_{i\in I(b_-)} (\bigvee_{k\in K(i)} \alpha^{k,i}_+), \bigwedge_{i\in I(b_-)}\gamma^i)= (\bigvee_k \alpha^{k, i(k)}_+, b_-) \in \conV. \]
    Because $\alpha^k_+ = \bigwedge_i \alpha^{k,i}_+ \leq \alpha^{k,i(k)}_+\ (\forall k)$, $x \leq \bigvee_k \alpha^k_+ \leq \bigvee_k \alpha^{k,i(k)}_+$ and so $(x,b_-)\in \downset \conV$.
\end{proof}

By Proposition~\ref{p:simple-axioms}, we know that $\bigoplus_i \p L^i$ is a d-frame and, moreover, by the same reasoning as for frames, we can prove that it has the universal property of a coproduct:
\begin{theorem}\label{t:coproduct-univ}
    $\bigoplus_i \p L^i$ is the coproduct in the category of d-frames.
\end{theorem}

\section*{Acknowledgement}
Discussions with Ale\v s Pultr helped greatly in the simplifications of Section~\ref{s:d-frames}.

\clearpage

\appendix

\section{Missing proofs in Section~\ref{s:presentations} and \ref{s:d-frames}}

\begin{proof}[Proof of Lemma~\ref{l:dfrm-univ} (Universality)]
    By Lemma~\ref{l:frm-univ}, $\sem$ is a pointwise meet-semilattice homomorphism which transforms $\p C_\pm$ covers into joins and, by the definition, also $\sem[\conO]\subseteq \con_*$ and $\sem[\totO]\subseteq \tot_*$. To show universality, let $f\colon \dB \to \p M$ be a presentation preserving map. Because the individual components of $f$, i.e. $f_\pm\colon B_\pm \to M_\pm$, are meet-preserving and transform covers into joins, by Lemma~\ref{l:frm-univ}, there exist unique frame homomorphisms $\overline f_\pm\colon L_\pm\to M_\pm$ such that $f_\pm = \overline f_\pm \circ \sem_\pm$.

    We need to prove that $\overline f = (\overline f_+, \overline f_-)$ is a d-frame homomorphism. Since $f$ preserves $\conO$ and $\totO$, we have that $\overline f[\sem [\conO]] = f[\conO] \subseteq \con_\p M$ which is equivalent to $\sem[\conO] \subseteq \overline f^{-1}[\con_\p M]$. Because $\overline f$ is a pair of frame homomorphisms, $\overline f^{-1}[\con_\p M]$ is $\downset$-closed, closed under $\wedge, \vee, \dirsqcup$ and contains $\dtt$ and $\dff$. However, \CON[{\sem[\conO]}] is the smallest such relation which contains $\sem[\conO]$ and so $\CON[{\sem[\conO]}] \subseteq \overline f^{-1}[\con_\p M]$. Finally, as before, this is equivalent to $ \overline f[\CON[{\sem[\conO]}]] \subseteq\con_\p M$. The case for $\overline f[\sem [\totO]] \subseteq \tot_\p M$ is similar.
\end{proof}

\begin{proof}[Proof of Lemma~\ref{l:indep-split}]
    The left-to-right implication holds immediately from (\ref{e:con-subsets}) on page~\pageref{e:con-subsets}  and the fact that \sCON{} is downwards closed. For the other implication, let $(b_+, b_-) \in (B_+\ttimes B_-)\cap \p D(\downset \conWV)$. By Lemma~\ref{l:d-equals-db}, there exist $A_\pm \subseteq B_\pm$ such that $b_\pm = \bigvee A_\pm$ and $A_+\ttimes A_- \subseteq \downset \conWV$. Fix an $a_-\in A_-$. Because $A_+\ttimes \{a_-\} \subseteq \downset \conWV$, for every $a^k_+\in A_+$, there exists an $\alpha^k\in \conWV$ such that $(a^k_+, a_-) \sqsubseteq \alpha^k$. Then,
    \[ (b_+,a_-) = (\bigvee A_+, a_-) \sqsubseteq (\bigvee_k \alpha^k_+, \bigwedge_k \alpha^k_-) \in \conWVbV = \conWbV \]
    and because $(b_+,a_-)\in B_+\ttimes B_-$ we can use ($\downset\conWV$-ind$_+$) and obtain that $(b_+, a_-)\in \downset \conWV$.

    Since $a_-\in A_-$ was chosen arbitrarily, $\{b_+\} \ttimes A_- \subseteq \downset \conWV$. Similarly to the above, for every $a^k_-\in A_-$ there is an $\alpha^k\in \conWV$ such that $(b_+, a^k_-) \sqsubseteq \alpha^k$, and $(b_+, b_-) \sqsubseteq (\bigwedge_k \alpha^k_+, \bigvee_k \alpha^k_-)\in \conVbW$. Finally, by ($\downset\conWV$-ind$_-$), $(b_+,b_-)\in \downset \conWV$.
\end{proof}

\begin{proof}[Proof of Proposition~\ref{p:simple-axioms}]
    We use Theorem~\ref{t:con-tot}. Clearly, (Indep$_\pm$) is a strengthening of ($\downset\conWV$-ind$_+$). To prove ($\lambda^4_\pm$-\con-\tot), let $\alpha \in \conWbV$ and $\beta\in \totW$ be such that $\beta_+ \leq \alpha_+$. Moreover, fix a $b_-\in B_-$ such that $b_- \leq \alpha_-$. Then, $(\alpha_+, b_-)\in \downset \conWbV$. By, (Indep$_+$), $(\alpha_+, b_-)\in \downset \conV$ and so there must be some $\gamma \in \conV$ such that $(\alpha_+, b_-)\sqsubseteq \gamma$. Because $\beta_+\leq \alpha_+ \leq \gamma_+$, by ($\mu_+$-\con-\tot), $b_- \leq \gamma_- \leq \beta_-$. Finally, because $b_-\in \downset \alpha_- \cap B_-$ was chosen arbitrarily, then also $\alpha_- = \bigvee (\downset \alpha_- \cap B_-) \leq \beta_-$.
\end{proof}

\section{Missing proofs in Section~\ref{s:coproducts}}

First, we prove general lemmas about a coproduct of frames $\bigoplus_i L^i$.

\begin{lemma}\label{p:copr-basics}
    Let $a,b\in L^j,\, \{a^k\}_k \subseteq L^j$ and $u, v\in B$. Then,
    \begin{enumerate}
        \item If $u\notin \mathbf n$, $\sem[u]\leq \sem[v]$ iff $u\leq v$.
        \item $\sem$ is injective on $B\setminus \mathbf n$.
        \item $(a \oplus_j u) \wedge (b \oplus_j u) = (a \wedge b) \oplus_j u$
        \item $\bigvee_k (a^k \oplus_j u) = (\bigvee_k a^k) \oplus_j u$
    \end{enumerate}
\end{lemma}
\begin{proof}
    (1) is exactly Proposition IV.5.2.4 in~\cite{picadopultr2011frames}. (2) follows from (1). For (3), recall that $\sem$ is a meet-semilattice homomorphism, so $(a \oplus_j u) \wedge (b \oplus_j u) = \sem[a *_j u] \wedge \sem[b *_j u] = \sem[(a *_j u) \wedge (b *_j u)] = \sem[(a\wedge b) *_j u] = (a\wedge b) \oplus_j u$. Finally, for (4), by Proposition IV.5.2.3 in~\cite{picadopultr2011frames}, we know that in a coproduct of two frames we have
    \[ \bigvee_k (a^k \oplus_1 (1,b)) = (\bigvee_k a^k) \oplus_1 (1,b) \]
    and we can view $\bigoplus_i L^i$ as a coproduct of two frames $L^j\oplus (\bigoplus_{i\in \p I\setminus \{j\}} L^i)$.
\end{proof}

\begin{lemma}[finite meets]\label{l:oplus-wedge}
    Let $\alpha^j = a^j \oplus_{i(j)} \ones$, for $j = 1, \dots, n$. Then,
    \[\bigwedge_{j=1}^n \alpha^j = \bigwedge_{i\in I} (b^i \oplus_i \ones) = \bigcap_{i\in I} (b^i \oplus_i \ones)\]
    where $I = \{ i(j) : j=1,\dots, n\}$, $b^i = \bigwedge \{ a^j ~|~ i(j) = i\})$.

    Moreover, $\bigwedge_{j=1}^n \alpha^j = \sem[u] = \downset u \cup \mathbf n$ where $u \in \prod'_i L^i$ such that $(u)_i = b^i$ for every $i\in I$ and $(u)_i = 1$ otherwise.
\end{lemma}
\begin{proof}
    By Lemma~\ref{p:copr-basics}, $\bigwedge_{j=1}^n \alpha^j = \bigwedge_{i\in I} (b^i \oplus_i \ones)$ and, because meets of \p C-ideals are computed as their intersection (see Section~\ref{s:frame-pres}), $\bigwedge_{j=1}^n \alpha^j = \bigcap_{i\in I} (b^i \oplus_i \ones)$. The `Moreover' part follows from this representation.
\end{proof}

\begin{lemma}[finite joins]\label{l:oplus-vee}
    Let $\alpha^j = a^j \oplus_{i(j)} \ones$, for $j = 1, \dots, n$. Then,
    \[\bigvee_{j=1}^n \alpha^j = \bigvee_{i\in I} (b^i \oplus_i \ones) = \bigcup_{i\in I} (b^i \oplus_i \ones)\]
    where $I = \{ i(j) : j=1,\dots, n\}$ and $b^i = \bigvee \{ a^j ~|~ i(j) = i\}$.
\end{lemma}
\begin{proof}
    Let $\beta^i \DEFEQ b^i \oplus_i \ones$, for every $i\in I$. First, we will show that $\bigcup_{i\in I} \beta^i$ is a $\p C$-ideal. Let $X = \{ x^k *_l u \}_{k\in K} \subseteq \bigcup_{i\in I} \beta^i$. Without loss of generality, assume that $X\cap \mathbf n = \emptyset$, i.e.\ $u_i \not= 0$, for all $i\not= l$, and that $x^k \not= 0$, for all $k\in K$. If $X\subseteq \beta^i$, for some $i$, then also $\bigvee_{k\in K} x^k *_l u \in \beta^i$ because $\beta^i$ is a $\p C$-ideal. Otherwise, there must exist an $m\in K$ which is different from $l$ such that $x^m *_l u \in \beta^m$. This means that $(x^m *_l u)_m = u_m \leq b^m$. From this we have that, for all $k\in K$, $(x^k *_l u)_m = u_m \leq b^m$ and, therefore, $x^k *_l u \in \beta^k$. Again, because $\beta^k$ is a $\p C$-ideal, $(\bigvee_{k\in K} x^k) *_l u \in \beta^k$.

    Next, $\alpha^j \subseteq \beta^{i(j)}$, for all $j=1,\dots, n$, and so $\bigcup_j \alpha^j \subseteq \bigcup_{i\in I} \beta^i$. Because $\bigvee_{j=1}^n \alpha^j$ is the smallest \p C-ideal containing $\bigcup_j \alpha^j$ and $\bigcup_{i\in I} \beta^i$ is also a \p C-ideal, we get that $\bigvee_{j=1}^n \alpha^j \subseteq \bigcup_{i\in I} \beta^j$. Finally, every $\beta^i \subseteq \bigvee \{\alpha^j ~|~ i(j) = i\} \subseteq \bigvee_{j=1}^n \alpha^j$ and so $\bigcup_{i\in I} \beta^j \subseteq \bigvee_{j=1}^n \alpha^j$.
\end{proof}

\bigskip

Now, let us look at the missing proofs from Section~\ref{s:coproducts}:

\begin{proof}[Proof of Lemma~\ref{p:contot-canon}]
    We prove that, for every $\alpha\in \conV$, $\alpha = \bigvee_{i\in I} \alpha^i$ for some $I\finsubseteq\p I$ and $\alpha^i = (b_+^i \oplus_i \ones,\, b_-^i \oplus_i \ones)\in  \conO$ (or $\totO$) and from this the lemma follows.

    Let $\bigvee_{j=1}^{n} \alpha^j\in \conV$ where $\alpha^j = (a_+^j \oplus_{i(j)} \ones,\, a_-^j \oplus_{i(j)} \ones)\in  \conO$, for all $j=1, \dots, n$. From Lemma~\ref{l:oplus-wedge} and Lemma~\ref{l:oplus-vee}, we have that
    \[ \bigvee_{j=1}^n \alpha^j = (\bigvee_{i\in I} (b_+^i \oplus_i \ones),\, \bigwedge_{i\in I} (b_-^i \oplus_i \ones)) \]
    where $I = \{ i(j) : j=1,\dots, n\}$ and $b^i_+ = \bigvee \{ a^j ~|~ i(j) = i\}$ and $b^i_- = \bigwedge \{ a^j ~|~ i(j) = i\}$. For every $i\in I$, because (\con-$\vee$) holds for $\p L^i$, $(b^i_+, b^i_-) \in \con^i$ and, therefore, also $(b_+^i \oplus_i \ones,\, b_-^i \oplus_i \ones) \in \conO$.
\end{proof}

\begin{proof}[Proof of Lemma~\ref{p:strips}]
    Recall that, for every $i\in \p I$, $\iota^i_\pm\colon L^i_\pm \to \bigoplus_i L^i_\pm$, $a \mapsto a \oplus_i \ones$, is a frame homomorphism obtained as a composition $\iota^i_\pm = \sem_\pm \circ \kappa^i_\pm$ (see Section~\ref{s:coprod-frm}). Here, $\kappa^i_\pm$ is always one-one and, because all $L^i_\pm$'s are nontrivial, $\kappa^i_\pm[L^i_\pm\setminus \{0\}] \cap \npm = \emptyset$ and so, by Lemma~\ref{p:copr-basics}, $\sem_\pm$ is one-one when restricted to image of $\kappa^i_\pm$. Therefore, $\iota^i$ is also one-one as it is a composition of those two maps. Moreover, $\iota^i$ is also onto $S^i_\pm$ and, because of Lemma~\ref{p:copr-basics},
    \[(a \oplus_i \ones) \wedge (b \oplus_i \ones) = (a\wedge b)\oplus_i \ones
        \qtq{and}
    \bigvee_k (a^k \oplus_i \ones) = (\bigvee_k a^k) \oplus_i \ones\]
    and so the associated inverse image map $(\iota^i_\pm)^{-1}\colon S^i_\pm \to L^i_\pm$ is also a frame homomorphism. Finally, from the definitions
    \begin{align*}
        (a,b)\in \con^i &\qtq{iff} (a \oplus_i \ones,\, b \oplus_i \ones) \in \conO \\
        (a,b)\in \tot^i &\qtq{iff} (a \oplus_i \ones,\, b \oplus_i \ones) \in \totO
    \end{align*}
    which makes $\iota^i$ restricted to the image, i.e.\ to $S^i_+\ttimes S^i_-$, an isomorphism of both structures.
\end{proof}

\begin{proof}[Proof of Theorem~\ref{t:coproduct-univ}]
    The last thing we need to check is universality of the construction. Let $\p M$ be a d-frame such that for every $j\in \p I$ there is a d-frame homomorphism $\lambda^j\colon \p L^j \to \p M$. We have the following situation:
    \begin{diagram}
        \p L^j
            \ar{r}{\kappa^j}
            \ar{rd}{\lambda^j} &
        (\prod'_i L^i_+, \prod'_i L^i_-, \conO, \totO)
            \ar{r}{\sem}
            \ar[dotted]{d}{\lambda} &
        \bigoplus_i \p L^i
            \ar[dotted]{ld}{\overline\lambda} \\

        & \p M
    \end{diagram}
    We obtain $\lambda = (\lambda_+, \lambda_-)$ from the universal property of the coproduct $\prod'_i L^i_\pm$: because all $\lambda^j_\pm$'s are meet-semilattice homomorphisms,  $\lambda_\pm$ is the unique meet-semilattice homomorphism such that $\lambda^j_\pm = \lambda_\pm \circ \kappa^j_\pm$. Since all $\lambda^j$ are d-frame homomorphisms, also $\lambda$ preserves $\conO$ and $\totO$ (follows immediately from the definition). It also preserves the cover relations in the presentation of d-frame coproducts (similarly to Proposition 4.3.2 in~\cite{picadopultr2011frames}):
    \begin{align*}
        \bigvee_k \lambda_\pm(a^k *_j u)
        &= \bigvee_k \lambda_\pm((a^k *_j \ones) \wedge (1 *_j u))
        \overset{(1)}= \bigvee_k (\lambda_\pm(a^k *_j \ones) \wedge \lambda_\pm(1 *_j u)) \\
        &= (\bigvee_k \lambda_\pm(a^k *_j \ones)) \wedge \lambda_\pm(1 *_j u) \\
        % = (\bigvee_k \kappa^j_\pm(a^k)) \wedge \lambda_\pm(1 *_j u) \\
        &\overset{(2)}= \lambda_\pm((\bigvee_k a^k)*_j \ones) \wedge \lambda_\pm(1 *_j u)
        \overset{(1)}= \lambda_\pm((\bigvee_k a^k)*_j u)
    \end{align*}
    Where (1)'s hold because $\lambda_\pm$ is a meet-semilattice homomorphism and (2) holds because $\lambda^j_\pm$ is a frame homomorphism and so $\bigvee_k \lambda_\pm(a^k *_j \ones) = \bigvee_k \lambda_\pm(\kappa^j_\pm(a^k)) = \bigvee_k \lambda^j_\pm(a^k) =  \lambda^j_\pm(\bigvee_k a^k) = \lambda_\pm(\kappa^j_\pm(\bigvee_k a^k)) = \lambda^j_\pm((\bigvee_k a^k)*_j \ones)$. Therefore, $\lambda$ preserves $\bigoplus_i \p L^i$ presentation and, by Lemma~\ref{l:dfrm-univ}, it can be uniquely extended to a d-frame homomorphisms $\overline\lambda$.
\end{proof}

\end{document}